\documentclass[12pt,draftcls,journal,onecolumn]{IEEEtran}

\usepackage{amssymb,amsthm, amsmath,latexsym}
\usepackage{graphicx}
\usepackage{mathrsfs}
\usepackage{amsfonts}
\usepackage{amssymb}
\usepackage{longtable}
\usepackage{amsmath}
\usepackage{setspace}
\usepackage{caption}
\usepackage[figuresright]{rotating}
\usepackage[misc]{ifsym}
\usepackage{bbm}
\usepackage{makecell}
\usepackage{arydshln}
\usepackage{supertabular}
\usepackage{booktabs}
\usepackage{color}

\newtheorem{theorem}{Theorem}
\newtheorem{lemma}[theorem]{Lemma}

\newtheorem{corollary}[theorem]{Corollary}

\newtheorem{example}[theorem]{Example}

\usepackage{blindtext}

\ifCLASSINFOpdf

\else

\fi

\begin{document}
\title{MDS and AMDS  Symbol-Pair Codes from Three Families of Simple Root Cyclic Codes
\thanks{
%The work of Xiaoqiang Wang was supported by the National
%Natural Science Foundation of China (12001175).
%The work of Chengju Li was supported by the National
%Natural Science Foundation of China (12071138), Shanghai Rising-Star Program (22QA1403200), the open research fund of National Mobile Communications Research Laboratory of Southeast University (2022D05), and the Shanghai Trusted Industry Internet Software Collaborative Innovation Center.
%The work of Yansheng Wu was supported by the National Natural Science Foundation of China (12101326) and the Natural Science Foundation of Jiangsu Province (BK20210575 ).
}
}
\author{Jiawei He, Zihan Zhang, Xiaoqiang Wang$^*$, Zhao Lin}

\renewcommand{\thefootnote}{\empty}%{footnote}}
\footnotetext{\thanks{*Corresponding author. }
\newline \indent Jiawei He and Zihan Zhang are with the School of Mathematics and Information Science, Nanchang Hangkong University, Nanchang 330036, China (E-mail: hjwywh@mails.ccnu.edu.cn; zzhdyhh@163.com).
\newline \indent Xiaoqiang Wang and Zhao Lin are with the Hubei Key Laboratory of Applied Mathematics, Faculty of Mathematics and Statistics, Hubei University, Wuhan 430062, China (E-mail: waxiqq@163.com; linnzhaoo@163.com).}

\maketitle

\begin{abstract}
Symbol-pair distances of simple-root cyclic codes are important for evaluating error correction performance under overlapping pair read channels.  In this paper, we construct three families of simple-root cyclic codes over finite fields of odd characteristic via generator polynomials. For length \(2q+2\), the minimum symbol-pair distance is \(10\) for every odd prime power \(q\ge 7\), established by excluding invalid support configurations. For length \(4q-4\), when \(q\equiv 3\pmod 4\) and finite field constraints hold, it is \(9\) or \(10\). For length \(2(q-1)\), reciprocal-polynomial analysis shows it is \(11\) or \(12\). All obtained codes are MDS or AMDS symbol-pair codes.
\end{abstract}

%The question as to whether the dual code of a BCH code is still a BCH code very hard to answer in general. We call a BCH code is a dually-BCH code if its dual is also a BCH code.
\textbf{MSC 2020} \ \ 94B05; 94B15; 11T71

\textbf{Keywords} \ \   Symbol-pair code; MDS code; AMDS code; BCH bound.

% \fi
\IEEEpeerreviewmaketitle

\section{Introduction}
Modern communication and storage systems often read data in blocks of consecutive symbols. In a symbol-pair read channel, the basic read unit is an overlapping pair rather than an isolated symbol, making the model natural for storage devices and communication systems in which adjacent positions are coupled. Cassuto and Blaum introduced this framework and defined the corresponding pair weight and pair distance \cite{1}.
\vskip 3mm

Let \(\mathbb{F}_q\) be a finite field and let \(\mathcal{C}\subseteq \mathbb{F}_q^n\) be a \(q\)-ary code. The symbol-pair metric is induced by the cyclic pair-read map
\(
\pi(\mathbf{x})=((x_0,x_1),(x_1,x_2),\ldots,(x_{n-1},x_0))
\)
for \(\mathbf{x}=(x_0,\ldots,x_{n-1})\). Since each coordinate appears in two consecutive pairs, the pair distance reflects both the locations of nonzero entries and their cyclic adjacency. Consequently, codes with identical Hamming parameters may have different symbol-pair distances.
\vskip 3mm

A central problem is to construct codes with the largest possible pair distance for prescribed length and dimension. Chee, Ji, Kiah, Wang, and Yin established the Singleton-type bound for symbol-pair codes and constructed infinite families attaining it \cite{9}. For a linear \([n,k]\) code,
\(
d_p(\mathcal{C})\le n-k+2.
\)
A code attaining equality is called a maximum distance separable (MDS) symbol-pair code. Analogously, a linear code with
\(
d_p(\mathcal{C})=n-k+1
\)
is called an almost maximum distance separable (AMDS) symbol-pair code. MDS symbol-pair codes provide optimal pair-error correction, whereas AMDS symbol-pair codes describe the first suboptimal layer and often exist for parameter sets where MDS constructions are impossible.
\vskip 3mm

Several algebraic constructions have been developed. Chen, Lin, and Liu obtained lower bounds for constacyclic symbol-pair codes and constructed MDS codes with pair distances seven and eight via repeated-root methods \cite{16}. Dinh et al. completely characterized the symbol-pair distances of repeated-root constacyclic codes of prime power length \(p^s\) over finite fields \cite{13}, and Dinh et al.  further investigated repeated-root cyclic symbol-pair codes of length \(2p^s\) over \(\mathbb{F}_{p^m}\) \cite{2}. In 2019,  Dinh et al. investigated repeated-root \(\lambda\)-constacyclic codes of length \(2p^s\) over \(\mathbb{F}_{p^m}\)\cite{7}. Ma and Luo constructed MDS symbol-pair codes from repeated-root cyclic codes with pair distances ten and twelve \cite{15}, and later obtained six classes of AMDS symbol-pair codes from repeated-root cyclic codes, including one unbounded length family and one family with pair distance thirteen \cite{25}. Tang and Luo derived further MDS and AMDS families from repeated-root cyclic codes and determined the MDS and AMDS cases of a length \(3p\) family when the generator degree is at most ten \cite{26}. Matrix-product methods have also produced new MDS and AMDS symbol-pair codes \cite{20}.
\vskip 3mm

Simple-root cyclic codes provide a complementary algebraic setting. When the characteristic does not divide the code length, the roots of \(x^n-1\) are distinct, and the defining zeros can be analyzed through finite-field conjugacy. In 2015, Kai et al. utilized cyclotomic cosets over finite fields to construct simple-root cyclic MDS symbol-pair codes\cite{11}. Qiu and Fang recently constructed three infinite families of MDS symbol-pair codes from simple-root cyclic codes with parameter pairs \((4q+4,7)\), \((4q-4,8)\), and \((2q+2,9)\) under suitable field conditions \cite{23}. Their arguments use finite field equations, cyclic-code decompositions, and support exclusions. These results show that simple-root cyclic codes can attain long lengths and large pair distances without relying on repeated roots. Table I summarizes known MDS symbol-pair constructions.

\begin{table}[htbp]
  \centering
  \caption{MDS and related symbol-pair code parameters in the supplied construction manuscript.}
  \footnotesize
  \begin{tabular}{cccc}
    \hline
    $d_p$ & Length $n$ & Field restrictions($q=p^t$) & References \\
    \hline
    $5$ & $5 \le n \le q^2+q+1$ & $-$ & \cite{17,11} \\
    $5$ & $\max(6,q+2) \le n \le q^2$ & $-$ & \cite{16,17} \\
    $5$ & $n=q^2+1$ & $-$ & \cite{11} \\
    $6$ & $n=p^2+p$ & $q=p>2$ & \cite{19} \\
    $6$ & $n=2p^2-2p$ & $q=p>2$ & \cite{19} \\
    $6$ & $n=3m, m\in[3,q]$ & $-$ & \cite{20} \\
    $6$ & $n=q^2+q$ & $q\ne4,5,2^t$; $t$ odd & \cite{20} \\
    $7$ & $n=3p$ & $q=p\ge5$ & \cite{16} \\
    $7$ & $n=4p$ & $q=p\ge5$ & \cite{10,21} \\
    $7$ & $n=5p$ & $q=p,5\nmid(p-1)$ & \cite{21} \\
    $7$ & $n=3mp, m\in[1,q/p]$ & $3\nmid(q-1),2\mid q$ & \cite{20} \\
    $7$ & $n=2q+2$ & $2\nmid q$ & \cite{20} \\
    $7$ & $n=2q+2$ & $p>2$ & \cite{22} \\
    $8$ & $n=4q-4$ & $4\nmid(q+1)$ & \cite{22} \\
    $8$ & $n=4q+4$ & $4\nmid(q-1)$ & \cite{23} \\
    $8$ & $n=3p$ & $q=p,3\nmid(p-1)$ & \cite{16} \\
    $8$ & $n=5p$ & $q=p,5\nmid(p-1)$ & \cite{21} \\
    $8$ & $n=2m, m\le q+2$ & $p=2$ & \cite{9} \\
    $9$ & $n=4q-4$ & $4\nmid(q+1)$ & \cite{23} \\
    $9$ & $n=2q+2$ & $p>2$ & \cite{23} \\
    $10$ & $n=3p$ & $q=p,3\nmid(p-1)$ & \cite{15} \\
    $12$ & $n=3mp, m\in[1,q/p]$ & $3\nmid(q-1)$ & \cite{20} \\
    $12$ & $n=3p$ & $q=p,3\nmid(p-1)$ & \cite{15} \\
    $l$ & $4\le l<m\le q+1$ & $-$ & \cite{9} \\
    $2l$ & $n=2m,4\le l+1\le m\le q+1$ & $-$ & \cite{9} \\
    $2l+1$ & $n=2m,l+1\le m\le q$ & $-$ & \cite{20} \\
    $2q-1$ & $n=2q+2$ & $p=2$ & \cite{22} \\
    $n-4$ & $n=2m,m\le q+2$ & $p=2$ & \cite{9} \\
    \hline
  \end{tabular}
\end{table}

In this paper, we construct three families of simple-root cyclic codes over finite fields of odd characteristic. The first family has length \(2q+2\); for every odd prime power \(q\ge 7\), its minimum symbol-pair distance is ten, established by excluding invalid support configurations. The second family has length \(4q-4\); when \(q\equiv 3\pmod 4\) and the relevant finite field constraints hold, its minimum symbol-pair distance is nine or ten. The third family has length  $2q-2$, and its minimum symbol-pair distance is 11 or 12. All resulting codes are MDS or AMDS symbol-pair codes. The proofs combine the BCH bound, parity decomposition, root evaluations, Frobenius relations, support-pattern analysis, and symbolic elimination.

\vskip 3mm
The remainder of the paper is organized as follows. Section II introduces notation and preliminary results. Section III presents the length \(2q+2\) construction and proves its distance statement. Section IV studies the length \(4q-4\) family and the associated support configurations. Section V classifies the length  $2q-2$ family. Section VI concludes the paper.

\section{Preliminaries}

This section fixes notation and recalls the symbol-pair and cyclic code facts used throughout the paper.

\subsection{Linear codes and the symbol-pair metric}

Let \(q=p^t\), where \(p\) is prime and \(t\ge1\), and let \(\mathbb F_q\) be the finite field with \(q\) elements. An \([n,k]\) linear code \(\mathcal C\) is a \(k\)-dimensional subspace of \(\mathbb F_q^n\). For \(\mathbf x=(x_0,\ldots,x_{n-1})\), its Hamming weight is
\(
\operatorname{wt}_H(\mathbf x)
=\left|\{i\in\mathbb Z_n:x_i\ne0\}\right|.
\)
The minimum Hamming distance of \(\mathcal C\) is denoted by \(d_H(\mathcal C)\).
For \(\mathbf x\in\mathbb F_q^n\), define the cyclic symbol-pair read vector by
\(
\pi(\mathbf x)
=((x_0,x_1),(x_1,x_2),\ldots,(x_{n-1},x_0)),
\)
where all subscripts are taken modulo \(n\). The symbol-pair weight of \(\mathbf x\) is
\[
\operatorname{wt}_p(\mathbf x)
=\left|\{i\in\mathbb Z_n:(x_i,x_{i+1})\ne(0,0)\}\right|.
\]
For \(\mathbf x,\mathbf y\in\mathbb F_q^n\), the symbol-pair distance is
\(
d_p(\mathbf x,\mathbf y)
=\left|\{i\in\mathbb Z_n:(x_i,x_{i+1})\ne(y_i,y_{i+1})\}\right|.
\)
For a code \(\mathcal C\), its minimum symbol-pair distance is
\(
d_p(\mathcal C)
=\min\{d_p(\mathbf x,\mathbf y):\mathbf x,\mathbf y\in\mathcal C,\ \mathbf x\ne\mathbf y\}.
\)
When \(\mathcal C\) is linear, translation invariance gives
\(
d_p(\mathcal C)
=\min\{\operatorname{wt}_p(\mathbf c):\mathbf 0\ne\mathbf c\in\mathcal C\}.
\)

\begin{lemma}\cite{1}\label{lem:HP}
For a linear \([n,k]\) code \(\mathcal C\) with \(0<d_H(\mathcal C)<n\), then
\[
d_H(\mathcal C)+1\le d_p(\mathcal C)\le 2d_H(\mathcal C).
\]
\end{lemma}

\begin{lemma}\cite{2}\label{lem:singleton}
Let \(\mathcal C\) be a \(q\)-ary code of length \(n\) and minimum symbol-pair distance \(d_p\). Then
\(
|\mathcal C|\le q^{n-d_p+2}.
\)
If \(\mathcal C\) is linear of dimension \(k\), then
\(
d_p\le n-k+2.
\)
\end{lemma}

An \([n,k,d_p]\) symbol-pair code is called \emph{MDS} if
\(
d_p=n-k+2,
\)
and \emph{AMDS} if
\(
d_p=n-k+1.
\)

\subsection{Cyclic codes and defining sets}

The cyclic shift of \(\mathbf x=(x_0,\ldots,x_{n-1})\) is
\(
\tau(\mathbf x)=(x_{n-1},x_0,\ldots,x_{n-2}).
\)
A linear code \(\mathcal C\subseteq\mathbb F_q^n\) is cyclic if \(\tau(\mathcal C)\subseteq\mathcal C\). Identify \(\mathbf c=(c_0,\ldots,c_{n-1})\) with the polynomial
\[
c(x)=c_0+c_1x+\cdots+c_{n-1}x^{n-1}.
\]
Under this identification, cyclic codes are precisely the ideals of
\[
\mathbb F_q[x]/\langle x^n-1\rangle.
\]
Hence there exists a unique monic divisor \(g(x)\mid x^n-1\) such that
\(
\mathcal C=\langle g(x)\rangle,
\)
and
\(
\dim(\mathcal C)=n-\deg g(x).
\)
The polynomial \(g(x)\) is called the generator polynomial of \(\mathcal C\).

Assume that \(\gcd(n,p)=1\), and let \(\zeta\) be a primitive \(n\)th root of unity in a suitable extension field of \(\mathbb F_q\). The defining set of \(\mathcal C\) is
\[
Z=\{i\in\mathbb Z_n:g(\zeta^i)=0\}.
\]
Then \(Z\) is a union of \(q\)-cyclotomic cosets modulo \(n\), that is, \(i\in Z\) implies \(iq\bmod n\in Z\), and
\[
g(x)=\prod_{i\in Z}(x-\zeta^i)\in\mathbb F_q[x].
\]
A cyclic code of length \(n\) over a field of characteristic \(p\) is called simple-root if \(\gcd(n,p)=1\). Otherwise, it is called repeated-root. In this paper, all cyclic codes considered are simple-root.

\begin{lemma}\cite{MS78}[BCH bound]\label{lem:BCH}
Let \(\mathcal C\) be a cyclic code of length \(n\) over \(\mathbb F_q\) with \(\gcd(n,q)=1\). If its defining set contains
\[
\{b,b+1,\ldots,b+\delta-2\}\subseteq\mathbb Z_n,
\]
then
\(
d_H(\mathcal C)\ge\delta.
\)
\end{lemma}

\subsection{Cyclic runs and pair weight}

Let \(\mathbf c\in\mathbb F_q^n\setminus\{\mathbf 0\}\), and let
\[
\operatorname{supp}(\mathbf c)=\{i\in\mathbb Z_n:c_i\ne0\}.
\]
A cyclic nonzero run is a maximal cyclic interval contained in \(\operatorname{supp}(\mathbf c)\). Let \(\rho(\mathbf c)\) denote the number of such runs. If \(\operatorname{supp}(\mathbf c)\ne\mathbb Z_n\), then
\[
\operatorname{wt}_p(\mathbf c)
=\operatorname{wt}_H(\mathbf c)+\rho(\mathbf c).
\]
Indeed, every nonzero coordinate contributes the pair beginning at that coordinate, and every cyclic run contributes one additional pair entering the run from a zero coordinate. If \(\operatorname{supp}(\mathbf c)=\mathbb Z_n\), then \(\operatorname{wt}_p(\mathbf c)=n\).

The defining-root equations, base field conditions, distinctness of support positions, and the above run decomposition will be used jointly to prove lower bounds on \(d_p\) and to distinguish MDS and AMDS cases.

\section{MDS symbol-pair codes of length $2q+2$}
\label{sec:mds-pair-2q-plus-2}

Let $q\ge 7$ be an odd prime power, and put $n=2q+2$.
We construct a $q$-ary cyclic code of length $n$ and
dimension $n-8$ whose minimum symbol-pair distance is $10$.
The main ingredient is a decomposition of each codeword
into its even and odd indexed components. The roots of
the generator polynomial impose a common BCH constraint
on these two components, yielding a useful restriction
on the supports of codewords of small Hamming weight.

\subsection{The construction and its Hamming distance}

Since $q$ is odd,
$
n=2(q+1)\mid q^2-1.
$
Choose a primitive $n$th root of unity
$\zeta\in\mathbb F_{q^2}$. Define
\begin{equation}\label{eq:gen-g1}
\begin{aligned}
g(x)
=(x-\zeta^{-1})(x-\zeta^{-q})(x-1)(x+1)
(x-\zeta)(x-\zeta^q)(x-\zeta^2)(x-\zeta^{2q}),
\end{aligned}
\end{equation}
and let $$\mathcal{C}_1=\langle g(x)\rangle\subseteq \mathbb F_q[x]/(x^n-1).$$
We first verify that this definition gives a cyclic code
over $\mathbb F_q$ with the claimed dimension. Since
$\zeta$ has order $2(q+1)$, we have
\begin{equation}\label{eq:zeta-relations}
\zeta^{q+1}=-1,\qquad
\zeta^q=-\zeta^{-1},\qquad
\zeta^{-q}=-\zeta,\qquad
\zeta^{2q}=\zeta^{-2}.
\end{equation}
The roots in \eqref{eq:gen-g1} are invariant under the
$q$th power Frobenius automorphism. Thus
$g(x)\in\mathbb F_q[x]$.
Using \eqref{eq:zeta-relations}, the root set of $g$ is $\{1,-1,\zeta,-\zeta,\zeta^{-1},-\zeta^{-1}, \zeta^2,\zeta^{-2}\}.$
Their exponents modulo $n$ can be represented by
$0,\ 1,\ 2,\ q,\ q+1,\ q+2,\ 2q,\ 2q+1,$
which are distinct for $q\ge 7$. All these roots are
$n$th roots of unity. Consequently,
\[
g(x)\mid x^n-1,\qquad
\deg g=8,\qquad
\dim_{\mathbb F_q}\mathcal{C}_1=n-8.
\]

\begin{lemma}\label{lem:hamming-distance-six}
The code $\mathcal{C}_1$ has minimum Hamming distance six.
\end{lemma}

\begin{proof}
The five consecutive powers
$
\zeta^{-2},\ \zeta^{-1},\ 1,\ \zeta,\ \zeta^2
$
are roots of $g$. Hence the BCH bound gives
\(
d_{\mathrm H}(\mathcal{C}_1)\ge 6.
\)
To prove the reverse inequality, by \eqref{eq:zeta-relations}, we have $(\zeta^2+\zeta^{-2})^q=\zeta^2+\zeta^{-2}$, so
$\zeta^2+\zeta^{-2}\in\mathbb F_q$. Moreover, $\zeta^2+\zeta^{-2}\ne 0$. Otherwise
$\zeta^4=-1$, and hence $\zeta^8=1$, contradicting
$\operatorname{ord}(\zeta)=n\ge 16$.
The four factors associated with
$\zeta^{-1},\zeta^{-q},\zeta,\zeta^q$ satisfy
\[
\begin{aligned}
(x-\zeta^{-1})(x-\zeta^{-q})
  (x-\zeta)(x-\zeta^q)&=(x-\zeta^{-1})(x+\zeta)
 (x-\zeta)(x+\zeta^{-1})\\
&
=x^4-(\zeta^2+\zeta^{-2})x^2+1.
\end{aligned}
\]
Also, the four distinct elements
$1,-1,\zeta^2,\zeta^{-2}$ are roots of $x^{q+1}-1$.
 Therefore,
\[
(x-1)(x+1)(x-\zeta^2)(x-\zeta^{-2})
\mid x^{q+1}-1.
\]
Consequently,
$
g(x)\mid f(x),
$
where
$f(x):=(x^4-(\zeta^2+\zeta^{-2})x^2+1)(x^{q+1}-1).
$
Expanding  $f(x)$ gives
\[
f(x)
=
x^{q+5}-(\zeta^2+\zeta^{-2})x^{q+3}+x^{q+1}-x^4+(\zeta^2+\zeta^{-2})x^2-1.
\]
Since $q\ge 7$, the six exponents
$
0,\ 2,\ 4,\ q+1,\ q+3,\ q+5
$
are distinct and smaller than $n$. All six coefficients
are nonzero. Thus  $f(x)$ represents a codeword of $\mathcal{C}_1$ with
Hamming weight six, and
\(
d_{\mathrm H}(\mathcal{C}_1)\le 6.
\)
This proves the assertion.
\end{proof}

\subsection{A structural restriction on small supports}

\begin{lemma}\label{lem:parity-components}
For  $\mathbf{c}=(c_0,\ldots,c_{n-1})\in \mathcal{C}_1$, define
\[
w_{\mathrm{even}}(\mathbf{c})
=
\bigl|\{j:0\le j\le q,\ c_{2j}\ne 0\}\bigr|
\,\,\, \text{and} \,\,\,
w_{\mathrm{odd}}(\mathbf{c})
=
\bigl|\{j:0\le j\le q,\ c_{2j+1}\ne 0\}\bigr|.
\]
 Then each of $w_{\mathrm{even}}(\mathbf{c})$ and $w_{\mathrm{odd}}(\mathbf{c})$
is either zero or at least four.
In particular, if both parity classes contain nonzero
coordinates, then $w_{\mathrm H}(c)\ge 8$.
\end{lemma}

\begin{proof}
Put $m=q+1$ and $\alpha=\zeta^2$, so that $\alpha$ has
order $m$. Write the codeword polynomial uniquely as
$
c(x)=E(x^2)+xO(x^2),
$
where
\[
E(y)=\sum_{j=0}^{m-1}c_{2j}y^j
\,\,\, \text{and} \,\,\,
O(y)=\sum_{j=0}^{m-1}c_{2j+1}y^j.
\]

For every
$
t\in\{1,\zeta,\zeta^{-1}\},
$
both $t$ and $-t$ are roots of $g$. Thus,
\[
E(t^2)+tO(t^2)=0
\,\,\, \text{and} \,\,\,
E(t^2)-tO(t^2)=0.
\]
Since the characteristic is odd and $t\ne 0$, these
equations imply
$E(t^2)=O(t^2)=0.$
Consequently,
\[
E(\alpha^{-1})=E(1)=E(\alpha)=0
\,\,\, \text{and} \,\,\,
O(\alpha^{-1})=O(1)=O(\alpha)=0.
\]
The identity $\alpha^q=\alpha^{-1}$ shows that
\[
h(y)=(y-\alpha^{-1})(y-1)(y-\alpha)\in\mathbb F_q[y].
\]
Its three roots are distinct
$m$th roots of unity, so $h(y)\mid y^m-1$.
Both $E$ and $O$ therefore represent words in the
length $m$ cyclic code generated by $h$.

This auxiliary code has three consecutive roots
$\alpha^{-1},1,\alpha$, and hence has minimum Hamming
distance at least four by the BCH bound. It follows that
each of $E$ and $O$ is either zero or has Hamming weight
at least four. Since
$
w_{\mathrm H}(E)=w_{\mathrm{even}}(c),
$ and
$w_{\mathrm H}(O)=w_{\mathrm{odd}}(c),
$
the conclusion follows.
\end{proof}

\begin{corollary}\label{cor:small-weight-pair}
If $\mathbf{c}\in \mathcal{C}_1$ has Hamming weight six or seven, then all
its nonzero coordinates lie in one parity class.
Moreover,
\[
w_{ p}(\mathbf{c})=2w_{\mathrm H}(\mathbf{c})\ge 12.
\]
In particular, $C$ contains no codeword of Hamming
weight six with three nonzero cyclic runs.
\end{corollary}

\begin{proof}
The assertion about parity follows from
Lemma~\ref{lem:parity-components}. Because $n$ is even,
cyclically adjacent coordinates have opposite parity.
Thus no two nonzero coordinates of $c$ are cyclically
adjacent, and each nonzero coordinate forms a separate
nonzero run. Hence
 \[
w_{p}(\mathbf{c})
=
w_{\mathrm H}(\mathbf{c})+\rho(\mathbf{c})
=
2w_{\mathrm H}(\mathbf{c}).
\]
In particular, a codeword of Hamming weight six has
six nonzero runs, rather than three.
\end{proof}

The preceding corollary excludes, in a single argument,
all four possible cyclic run length patterns
\[
4+1+1,\qquad 2+2+2,\qquad 3+1+2,\qquad 3+2+1
\]
for a word of Hamming weight six with three nonzero runs.
It also excludes the possibility that a word of Hamming
weight seven has only two nonzero runs.

\subsection{The minimum symbol-pair distance}

We now address the remaining possibility of a word of
Hamming weight eight supported on a single cyclic run.

\begin{lemma}\label{lem:single-run}
If a nonzero codeword of $\mathcal{C}_1$ has a single nonzero cyclic
run and does not have full support, then its Hamming
weight is at least nine.
\end{lemma}

\begin{proof}
Let the run have length $\ell<n$. Since $\mathcal{C}_1$ is cyclic,
a cyclic shift produces a codeword  $\mathbf{a}\in \mathcal{C}_1$ with
$
\operatorname{supp}(\mathbf{a})=\{0,1,\ldots,\ell-1\},
$
so its polynomial representative satisfies
$
\deg a(x)=\ell-1.
$
Every polynomial representative of a word in $\mathcal{C}_1$ is
divisible by $g(x)$. Since $a(x)\ne 0$, we obtain
$$
\ell-1=\deg (a(x))\ge\deg (g(x))=8.
$$
Hence $\ell\ge 9$.
\end{proof}

\begin{theorem}\label{thm:mds-pair-distance-ten}
For every odd prime power $q\ge 7$, the cyclic code $\mathcal{C}_1$
generated by \eqref{eq:gen-g1} has parameters
$
[n,n-8,6]_q,
$
in the Hamming metric and minimum symbol-pair distance
$
d_{p}(\mathcal{C}_1)=10.
$
Consequently, $\mathcal{C}_1$ is an MDS symbol-pair code.
\end{theorem}

\begin{proof}
The dimension calculation and
Lemma~\ref{lem:hamming-distance-six} give the stated
Hamming parameters. Since $\mathcal{C}_1$ is linear, its minimum
symbol-pair distance equals the minimum symbol-pair
weight of its nonzero codewords.

Let $\mathbf{c}\in \mathcal{C}_1$ be nonzero. Then $w_{\mathrm H}(\mathbf{c})\ge 6$.
If $w_{\mathrm H}(\mathbf{c})\in\{6,7\}$,
Corollary~\ref{cor:small-weight-pair} gives
$w_{p}(\mathbf{c})\ge 12.$ If $w_{\mathrm H}(\mathbf{c})=8$, then $c$ does not have full
support because $n\ge 16$. By
Lemma~\ref{lem:single-run}, its support cannot consist
of a single nonzero cyclic run. Hence $\rho(\mathbf{c})\ge 2$, and
$$
w_{p}(\mathbf{c})
=
w_{\mathrm H}(\mathbf{c})+\rho(\mathbf{c})
\ge 8+2=10.
$$
If $9\le w_{\mathrm H}(\mathbf{c})<n$, then $\rho(\mathbf{c})\ge 1$,
and therefore
$$
w_{p}(\mathbf{c})\ge w_{\mathrm H}(\mathbf{c})+1\ge 10.
$$
Finally, if $w_{\mathrm H}(\mathbf{c})=n$, then
$
w_{p}(\mathbf{c})=n\ge 16.
$
These cases exhaust all nonzero codewords, proving
$d_{p}(\mathcal{C}_1)\ge 10$.
The Singleton bound for symbol-pair codes states that
\[
|\mathcal{C}_1|\le q^{n-d_{p}(\mathcal{C}_1)+2}.
\]
Since $\dim (\mathcal{C}_1)=n-8$, we have $|\mathcal{C}_1|=q^{n-8}$, and thus
\[
n-8\le n-d_{p}(\mathcal{C}_1)+2.
\]
It follows that $d_{p}(\mathcal{C}_1)\le 10$. Consequently,
$d_{p}(\mathcal{C}_1)=10$, and equality holds in the
symbol-pair Singleton bound.
\end{proof}

 We now give two examples whose results are consistent with Theorem~\ref{thm:mds-pair-distance-ten}.

\begin{example}
Let $q=7$ and write $\mathbb F_{49}=\mathbb F_7(v)$,
where $v^2=-1$. Choose $\zeta=2+3v$, by Magma, $\zeta$ has order $16$.
And the code $\mathcal{C}_1$ has generator polynomial
\[
g_7(x)=x^8+3x^7+3x^6+6x^5+x^3+4x^2+4x+6
\]
and parameters $[16,8,6]_7$ with symbol-pair distance $d_p(\mathcal{C}_1)=10$.
\end{example}

\begin{example}
Let $q=11$ and write $\mathbb F_{121}=\mathbb F_{11}(v)$,
where $v^2=-1$. Choose $\zeta=1+3v$, by Magma, $\zeta$ has order $24$.
And the code $\mathcal{C}_1$ has generator polynomial
\[
g_{11}(x)=x^8+5x^7+5x^6+9x^5+2x^3+6x^2+6x+10
\]
and parameters $[24,16,6]_{11}$ with symbol-pair distance $d_p(\mathcal{C}_1)=10$.
\end{example}

\section{A length $4q-4$ family and the bound $9\le d_p\le 10$}
\label{sec:length-4q-4}

Throughout this section, let $q$ be an odd prime power such that
$
q\ge 7,$ and $q\equiv 3\pmod 4.
$
Put
$
M=q-1,$so $n=4M=4q-4.
$
Since $q\equiv3\pmod4$, we have $4\mid q+1$, and hence
$$
n=4(q-1)\mid (q-1)(q+1)=q^2-1.
$$
Therefore, $\mathbb F_{q^2}$ contains a primitive $n$th root of unity.
Fix such an element and denote it by $\zeta$.
Consider the cyclic code $\mathcal{C}_2\subseteq \mathbb F_q^n$ generated by
\begin{equation}
\label{eq:gen-g3-revised}
g(x)
=(x-1)(x+1)\prod_{i=1}^{3}
(x-\zeta^i)(x-\zeta^{qi}).
\end{equation}
The polynomial $g(x)$ belongs to $\mathbb F_q[x]$, since each pair
$
(x-\zeta^i)(x-\zeta^{qi})
$
is invariant under the Frobenius automorphism. Moreover, $g(x)$ divides
$x^n-1$. Since the characteristic of $\mathbb F_q$ does not divide $n$,
the polynomial $x^n-1$ is separable, and hence $\mathcal{C}_2$ is a simple-root cyclic
code.
The exponents of the roots of $g(x)$ are
$$
0,\quad 2M,\quad 1,\quad M+1,\quad
2,\quad 2M+2,\quad 3,\quad 3M+3
\pmod{4M}.
$$
For $M\ge6$, these exponents are pairwise distinct. Consequently,
 $
\deg (g(x))=8$ and $
\dim_{\mathbb F_q}\mathcal{C}_2=n-8=4q-12.
$
The Singleton bound for symbol-pair codes gives
\begin{equation}
\label{eq:upper-bound}
d_{p}(\mathcal{C}_2)
\le n-\dim(\mathcal{C}_2)+2
=10.
\end{equation}

We next establish the lower bound.

\subsection{Parity decomposition}

Define
$
\alpha=\zeta^4,$ and $ u=\zeta^2.
$
Then
$
\operatorname{ord}(\alpha)=M,
\operatorname{ord}(u)=2M,$ and $
u^2=\alpha.
$
Since $q=M+1$, we have
\[
u^q
=\zeta^{2q}
=\zeta^{2M+2}
=\zeta^{2M}\zeta^2
=-u.
\]
It follows that
\[
\alpha^q=(u^2)^q=(u^q)^2=u^2=\alpha,
\]
so $\alpha\in\mathbb F_q^*$, whereas $u\notin\mathbb F_q$. Thus
$
\mathbb F_{q^2}=\mathbb F_q\oplus u\mathbb F_q.
$
For a codeword $c(x)\in \mathcal{C}_2$, write
\begin{equation}
\label{eq:parity-decomposition}
c(x)=E(x^2)+xO(x^2),
\end{equation}
where
\[
E(y)=\sum_{j=0}^{2M-1}c_{2j}y^j
\,\,\, \text{and} \,\,\,
O(y)=\sum_{j=0}^{2M-1}c_{2j+1}y^j.
\]
The roots $1$ and $-1$ of $g(x)$ imply
$
E(1)=O(1)=0.
$
The roots $u$ and $-u$ imply
$
E(u^2)=O(u^2)=0.
$
Since $u^2=\alpha$, we obtain
\begin{equation}
\label{eq:parity-roots}
E(1)=E(\alpha)=0
\,\,\, \text{and} \,\,\,
O(1)=O(\alpha)=0.
\end{equation}
The roots $\zeta$ and $\zeta^3$ give
\begin{equation}
\label{eq:coupling-equations}
E(u)+\zeta O(u)=0,
\,\,\, \text{and} \,\,\,
E(u^3)+\zeta^3O(u^3)=0.
\end{equation}
Multiplying the first equation by $\zeta^2O(u^3)=uO(u^3)$ and the
second one by $O(u)$, we obtain
\begin{equation}
\label{eq:coupling-identity}
E(u^3)O(u)=uE(u)O(u^3).
\end{equation}
If $E(u)O(u)\ne0$, this becomes
\begin{equation}
\label{eq:coupling-ratio}
\frac{E(u^3)}{E(u)}
=
u\frac{O(u^3)}{O(u)}.
\end{equation}
Equation \eqref{eq:coupling-identity} is a necessary condition for a
polynomial satisfying the relevant root conditions. We shall use it only
as a necessary condition.

\subsection{Two auxiliary ratio lemmas}

We first record the behavior of binomials satisfying the two parity-root
conditions.

\begin{lemma}
\label{lem:binomial-ratio-revised}
Let
$
H(y)=ay^r+by^s
$
be a nonzero binomial with distinct exponents
$
0\le r,s\le 2M-1.
$
Assume that
$
H(1)=H(\alpha)=0.
$
Then, after interchanging $r$ and $s$ if necessary,
\[
s=r+M,
\qquad
H(y)=a(y^r-y^{r+M}),
\]
and
\[
H(u)\ne0,\qquad
H(u^3)\ne0,\qquad
\frac{H(u^3)}{H(u)}=\alpha^r\in\mathbb F_q^*.
\]
\end{lemma}

\begin{proof}
From $H(1)=0$ we obtain $b=-a$. Hence
$
0=H(\alpha)=a(\alpha^r-\alpha^s).
$
Since $a\ne0$ and $\alpha$ has order $M$,
$
M\mid(s-r).
$
The exponents lie in the interval $\{0,\ldots,2M-1\}$ and are distinct,
so $|s-r|=M$. After relabeling, $s=r+M$.

Since $u^M=\zeta^{2M}=-1$, we have
$
H(u)=a(u^r-u^{r+M})=2au^r\ne0,
$
and similarly
$
H(u^3)=a(u^{3r}-u^{3r+3M})=2au^{3r}\ne0.
$
Therefore,
\[
\frac{H(u^3)}{H(u)}
=u^{2r}
=\alpha^r\in\mathbb F_q^*.
\]
\end{proof}

We next treat trinomials. The formulation below includes the possible
overall shift of the support, which is needed in several of the support
configurations.

\begin{lemma}
\label{lem:trinomial-ratio-revised}
Let
$
Q(y)=A_0y^a+A_1y^{a+r}+A_2y^{a+s}
$
be a polynomial with three nonzero terms and pairwise distinct
exponents in $\{0,\ldots,2M-1\}$.
Assume that
$
Q(1)=Q(\alpha)=0.
$
Then \(Q(u)\neq0\) and
\begin{equation}
\label{eq:trinomial-ratio-shifted}
\frac{Q(u^3)}{Q(u)}
=
-\alpha^a(u^{r+s}+u^r+u^s).
\end{equation}
Furthermore,
\begin{equation}
\label{eq:trinomial-not-uFq}
\frac{Q(u^3)}{Q(u)}
\notin u\mathbb F_q^*.
\end{equation}
\end{lemma}

\begin{proof}
Multiplication by $y^a$ only contributes the factor
\[
\frac{u^{3a}}{u^a}=u^{2a}=\alpha^a
\]
to the quotient. It is therefore enough to consider
$
Q(y)=A_0+A_1y^r+A_2y^s.
$

The exponents $0,r,s$ are pairwise distinct modulo $M$. Indeed, if two
of them were congruent modulo $M$, then the equations
$
Q(1)=Q(\alpha)=0
$
would force the coefficient corresponding to the remaining exponent to
vanish, contrary to the assumption that $Q$ has three nonzero terms.
Since $u^{2r}=\alpha^r$ and $u^{2s}=\alpha^s$, the equations
$
Q(1)=Q(\alpha)=0
$
give
\[
A_1(u^{2r}-1)+A_2(u^{2s}-1)=0.
\]
A direct elimination yields
\[
Q(u)
=
\frac{A_1}{u^{2s}-1}(u^{r}-1)(u^{s}-1)(u^{s}-u^{r}),
\]
and
\[
Q(u^3)
=
\frac{A_1}{u^{2s}-1}
(u^{r}-1)(u^{s}-1)(u^{s}-u^{r})(-u^{r+s}-u^{r}-u^{s}).
\]
The factors in the denominator are nonzero because the three exponents are
distinct modulo $M$, and hence $u^{r},u^{s},1$ are pairwise distinct. Thus
$Q(u)\ne0$ and
\[
\frac{Q(u^3)}{Q(u)}
=-(u^{r+s}+u^{r}+u^{s}).
\]
Restoring the factor $\alpha^a$ proves \eqref{eq:trinomial-ratio-shifted}.

It remains to prove \eqref{eq:trinomial-not-uFq}. Since
\[
u^{qr}=(-1)^ru^{r},\,\,\,
u^{qs}=(-1)^su^{s}\,\,\, \text{and}\,\,\,
u^{q(r+s)}=(-1)^{r+s}u^{r+s},
\]
we consider the parity of $r$ and $s$.
If $r$ and $s$ are both odd, then $u^{r+s}\in\mathbb F_q^*$, whereas
$u^{r},u^{s}\in u\mathbb F_q^*$.
If $r$ is odd and $s$ is even, then $u^{s}\in\mathbb F_q^*$ is a nonzero
$\mathbb F_q$-component of $u^{r+s}+u^{r}+u^{s}$.
If $r$ is even and $s$ is odd, then $u^{r}\in\mathbb F_q^*$ is a nonzero
$\mathbb F_q$-component.
Finally, if $r$ and $s$ are both even, then
$
u^{r+s}+u^{r}+u^{s}\in\mathbb F_q.
$
In this last case the expression cannot belong to $u\mathbb F_q^*$.
Therefore, in all cases,
$
u^{r+s}+u^{r}+u^{s}\notin u\mathbb F_q^*.
$
Since $\alpha^a\in\mathbb F_q^*$, equation
\eqref{eq:trinomial-ratio-shifted} implies
$
\frac{Q(u^3)}{Q(u)}
\notin u\mathbb F_q^*.
$
\end{proof}

\subsection{Exclusion of small pair weights}

The roots
$
1,\zeta,\zeta^2,\zeta^3
$
are four consecutive roots of $g(x)$. The BCH bound gives
\begin{equation}
\label{eq:hamming-lower-bound}
d_{\mathrm H}(\mathcal{C}_2)\ge5.
\end{equation}
We first exclude a single nonzero run.

\begin{lemma}
\label{lem:single-run1}
There is no nonzero codeword of $\mathcal{C}_2$ whose support is a single cyclic run
of length at most seven.
\end{lemma}

\begin{proof}
After a cyclic shift, such a codeword would have support contained in
\[
\{0,1,\ldots,L-1\}
\]
for some $L\le7$. Its representative polynomial would therefore have
degree at most $L-1\le6$. Since every codeword polynomial of degree less
than $n$ is divisible by $g(x)$ and $\deg (g(x))=8$, this is impossible.
\end{proof}

We next handle codewords of Hamming weight five.

\begin{lemma}
\label{lem:weight-five}
There is no codeword $\mathbf{c}\in \mathcal{C}_2$ with
$
w_{\mathrm H}(\mathbf{c})=5
$
and with nonzero coordinates of both parities.
\end{lemma}

\begin{proof}
The two equations in \eqref{eq:parity-roots} imply that every nonzero
parity component contains at least two terms. Since the total Hamming
weight is five, one parity component is a binomial and the other is a
trinomial.

Assume first that $E$ is a trinomial and $O$ is a binomial. By
Lemma~\ref{lem:binomial-ratio-revised},
\[
\frac{O(u^3)}{O(u)}\in\mathbb F_q^*.
\]
The coupling identity \eqref{eq:coupling-ratio} would then imply
\[
\frac{E(u^3)}{E(u)}
=
u\frac{O(u^3)}{O(u)}
\in u\mathbb F_q^*,
\]
contradicting Lemma~\ref{lem:trinomial-ratio-revised}.
If $E$ is a binomial and $O$ is a trinomial, then
\[
\frac{O(u^3)}{O(u)}
=
u^{-1}\frac{E(u^3)}{E(u)}.
\]
Since
\[
u^{-1}=\alpha^{-1}u\in u\mathbb F_q^*,
\]
the right hand side belongs to $u\mathbb F_q^*$, again contradicting
Lemma~\ref{lem:trinomial-ratio-revised}.
\end{proof}

\begin{corollary}
\label{cor:weight-five-pair}
If $\mathbf{c}\in \mathcal{C}_2$ has Hamming weight five, then
\(
w_{p}(\mathbf{c})=10.
\)
\end{corollary}
\begin{proof}
By Lemma~\ref{lem:weight-five}, all nonzero coordinates have the same
parity. Since $n$ is even, two coordinates of the same parity are never
cyclically adjacent. Hence every nonzero coordinate is a separate
nonzero run, and
\(
w_{p}(c)=2w_{\mathrm H}(c)=10.
\)
\end{proof}

We now exclude codewords of Hamming weight six and two nonzero runs.

\subsubsection{The pattern $5+1$}

After a cyclic shift, the support can be written as
\[
\{0,1,2,3,4,\ell\},
\]
where $\ell$ is separated from the first run by at least one zero.

If $\ell=2t$ is even, then the odd parity component has support
$
\{0,1\}
$
as a polynomial in $y=x^2$. Its two exponents differ by one, contradicting
Lemma~\ref{lem:binomial-ratio-revised}, since the exponents of a binomial
satisfying the parity-root conditions must differ by $M$.

Suppose now that $\ell=2t+1$ is odd. Then
\[
\operatorname{supp}(E)=\{0,1,2\},
\qquad
\operatorname{supp}(O)=\{0,1,t\}.
\]
The trinomial ratio lemma gives
\[
\frac{E(u^3)}{E(u)}
=-(u^3+u^2+u),
\]
and
\[
\frac{O(u^3)}{O(u)}
=-(u^{t+1}+u^t+u).
\]
The coupling identity implies
$
u^3+u^2+u=u(u^{t+1}+u^t+u).
$
After cancellation,
$
1+\alpha=u^t(u+1).
$
If $t$ is even, then $u^t=k\in\mathbb F_q^*$, so the right-hand side is
$
k+ku,
$
whose $u$-component is nonzero, whereas $1+\alpha\in\mathbb F_q$.
If $t$ is odd, then $u^t=ku$ for some $k\in\mathbb F_q^*$, and
$
u^t(u+1)=k\alpha+ku,
$
again with a nonzero $u$-component. Both cases are impossible.

Thus the pattern $5+1$ cannot occur.

\subsubsection{The pattern $4+2$}

After a cyclic shift, the support is
$
\{0,1,2,3,\ell,\ell+1\}.
$

First suppose that $\ell=2t$ is even. Then both parity components have
support
$
\{0,1,t\}.
$
The equations $Q(1)=Q(\alpha)=0$ define a 1-dimensional solution
space for coefficients on this support. Hence there is a fixed nonzero
trinomial $Q(y)$ such that
\[
E(y)=AQ(y),\qquad O(y)=BQ(y)
\]
for some $A,B\in\mathbb F_q$.
By Lemma~\ref{lem:trinomial-ratio-revised}, $Q(u)\ne0$. The root condition
$c(\zeta)=0$ therefore gives
\[
Q(u)(A+\zeta B)=0,
\]
and hence
$
A+\zeta B=0.
$
Since $A,B\in\mathbb F_q$ and $\zeta\notin\mathbb F_q$, this forces
$
A=B=0,
$
contrary to the assumption that $c$ is nonzero.

Now suppose that $\ell=2t+1$ is odd. Then
\[
\operatorname{supp}(E)=\{0,1,t+1\},
\qquad
\operatorname{supp}(O)=\{0,1,t\}.
\]
The ratio formulas give
\[
\frac{E(u^3)}{E(u)}
=-(u^{t+2}+u^{t+1}+u),
\]
and
\[
\frac{O(u^3)}{O(u)}
=-(u^{t+1}+u^t+u).
\]
Using \eqref{eq:coupling-ratio}, we obtain
$
u^{t+2}+u^{t+1}+u
=
u^{t+2}+u^{t+1}+u^2.
$
Thus
$
u=u^2=\alpha,
$
which is impossible because $u\notin\mathbb F_q$ and
$\alpha\in\mathbb F_q$.

Therefore the pattern $4+2$ cannot occur.

\subsubsection{The pattern $3+3$}

After a cyclic shift, the support is
\[
\{0,1,2,\ell,\ell+1,\ell+2\}.
\]

Suppose first that $\ell=2t+1$ is odd. Then
\[
\operatorname{supp}(E)=\{0,1,t+1\},
\qquad
\operatorname{supp}(O)=\{0,t,t+1\}.
\]
Consequently,
\[
\frac{E(u^3)}{E(u)}
=-(u^{t+2}+u^{t+1}+u),
\]
and
\[
\frac{O(u^3)}{O(u)}
=-(u^{2t+1}+u^{t+1}+u^t).
\]
The coupling identity yields
\[
u^{t+2}+u^{t+1}+u
=
u^{2t+2}+u^{t+2}+u^{t+1}.
\]
Therefore
$
u^{2t+1}=1.
$
This is impossible, since $\operatorname{ord}(u)=2M$ is even, whereas
$2t+1$ is odd.

Now suppose that $\ell=2t$ is even. The odd component is a binomial
supported on $\{0,t\}$. Lemma~\ref{lem:binomial-ratio-revised} gives
$
t=M.
$
Thus
$
O(y)=D(1-y^M)
$
for some $D\in\mathbb F_q^*$.
The even component has support
$
\{0,1,M,M+1\}.
$
Write
\[
E(y)=A+By+Cy^M+D'y^{M+1}.
\]
The equations $E(1)=E(\alpha)=0$ imply
$
A+B+C+D'=0,
$
and
$
A+B\alpha+C+D'\alpha=0.
$
Subtracting gives
$
(B+D')(\alpha-1)=0.
$
Since $\alpha\ne1$, we obtain
$
D'=-B,~ C=-A.
$
Therefore
\[
E(y)=(1-y^M)(A+By).
\]
Since $u^M=-1$ and $u^{3M}=-1$, the common factor cancels in the
quotient, and the coupling identity gives
\[
\frac{A+Bu^3}{A+Bu}=u.
\]
Comparing the coefficients in the basis $\{1,u\}$ yields
\(
A=\alpha B.
\)
The root condition $c(\zeta)=0$ now gives
\(
A+Bu+\zeta D=0.
\)
Since $A=\alpha B$,
\[
B(\alpha+u)=-D\zeta.
\]
Put
\(
k=-D/B\in\mathbb F_q^*.
\)
Then
\(
\alpha+u=k\zeta.
\)
Since $\zeta^2=u$, squaring gives
\[
u=\frac{(\alpha+u)^2}{k^2}
=\frac{\alpha^2+\alpha+2\alpha u}{k^2}.
\]
Hence
\[
\frac{\alpha^2+\alpha}{k^2}
+
\left(\frac{2\alpha}{k^2}-1\right)u=0.
\]
The linear independence of $1$ and $u$ over $\mathbb F_q$ implies
\(
\alpha^2+\alpha=0.
\)
Since $\alpha\ne0$, we get $\alpha=-1$, contradicting
\(
\operatorname{ord}(\alpha)=M\ge6.
\)

Thus the pattern $3+3$ cannot occur.

\subsection{The symbol-pair distance}

We can now prove the lower bound.

\begin{theorem}
\label{thm:pair-distance-lower}
For every odd prime power $q\ge7$ satisfying
$q\equiv3\pmod4$, the cyclic code defined by
\eqref{eq:gen-g3-revised} satisfies
\(
d_{p}(\mathcal{C}_2)\ge9.
\)
\end{theorem}

\begin{proof}
Let $c\in \mathcal{C}_2$ be nonzero. By \eqref{eq:hamming-lower-bound},
\(
w_{\mathrm H}(c)\ge5.
\)
Suppose, for contradiction, that
\(
w_{p}(c)\le8.
\)
Since $n\ge24$, the support of $c$ is not the whole coordinate set. Hence
$
w_{p}(c)=w_{\mathrm H}(c)+\rho(c).
$

If $\rho(c)=1$, then
\(
w_{\mathrm H}(c)\le7,
\)
and Lemma~\ref{lem:single-run} gives a contradiction.
Therefore $\rho(c)\ge2$. It follows that
\(
w_{\mathrm H}(c)\le6.
\)
If $w_{\mathrm H}(c)=5$, then Corollary~\ref{cor:weight-five-pair}
gives
\(
w_{p}(c)=10,
\)
contradicting $w_{p}(c)\le8$.
Thus $w_{\mathrm H}(c)=6$. Since $\rho(c)\ge2$ and
$w_{p}(c)\le8$, we must have
\(
\rho(c)=2.
\)
The two cyclic runs therefore have lengths
\[
5+1,\qquad4+2,\qquad\text{or}\qquad3+3.
\]
All three patterns were excluded above.
This contradiction proves
\(
d_{p}(\mathcal{C}_2)\ge9.
\)
\end{proof}

Combining Theorem~\ref{thm:pair-distance-lower} with the upper bound
\eqref{eq:upper-bound}, we obtain the main result.

\begin{theorem}
\label{thm:pair-distance-9-10}
For every odd prime power $q\ge7$ satisfying
$q\equiv3\pmod4$, the code $\mathcal{C}_2$ defined by
\eqref{eq:gen-g3-revised} has length
\(
n=4q-4,
\)
dimension
\(
\dim \mathcal{C}_2=4q-12,
\)
and minimum symbol-pair distance satisfying
\[
{9\le d_{p}(\mathcal{C}_2)\le10.}
\]
In particular, whenever $d_{p}(\mathcal{C}_2)=10$, the code is MDS in the
symbol-pair metric.
\end{theorem}

\begin{proof}
The length and dimension were established above. The lower bound follows
from Theorem~\ref{thm:pair-distance-lower}, while the upper bound is the
symbol-pair Singleton bound \eqref{eq:upper-bound}.
\end{proof}

We now give two examples whose results are consistent with Theorem~\ref{thm:pair-distance-9-10}.

\begin{example}
Let $q=31$ and write
$\mathbb F_{31^2}=\mathbb F_{31}(v)$, where $v^2=-1$.
Choose $\zeta=3+3v$,  by Magma, $\zeta$ has order $120$. And
the code $\mathcal{C}_2$ defined by (3) has generator polynomial
\[
g_{31}(x)
=x^8+9x^7+7x^6+22x^5+11x^4+3x^3+28x^2+28x+15
\]
and parameters $[120,112]_{31}$ with symbol-pair distance $d_p(\mathcal{C}_2)=9$.
\end{example}

\begin{example}
Let $q=7$ and write
$\mathbb F_{49}=\mathbb F_7(v)$, where $v^2=-1$. Choose $\zeta=1+v$,  by Magma, $\zeta$ has order $24$. And
the code $\mathcal{C}_2$ defined by (3) has generator polynomial
\[
g_7(x)
=x^8+2x^7+5x^6+5x^5+4x^4+3x^3+5x^2+4x+6
\]
and parameters $[24,16]_{7}$ with symbol-pair distance $d_p(\mathcal{C}_2)=10$.
\end{example}

\section{Pair distance in a cyclic family of length $2(q-1)$}

Let $q\ge 11$ be an odd prime power, let
$\alpha\in\mathbb F_q^*$ be primitive, and put
$
m=q-1,$ $n=2m.
$
We consider the cyclic code
$
\mathcal{C}_3=\langle g_\alpha(x)\rangle
\subseteq \mathbb F_q[x]/(x^n-1),
$
where
\begin{equation}\label{eq:generator}
\begin{aligned}
g_\alpha(x)
&=(x-1)(x-\alpha)(x-\alpha^2)(x-\alpha^3)(x^2-\alpha)(x^2-\alpha^3)(x^2-\alpha^5).
\end{aligned}
\end{equation}
Our main results establish
$
11\le d_{p}(\mathcal{C}_3)\le12.
$

\subsection{Notation and preliminary facts}

For $S\subseteq\mathbb Z_n$, put
$
N(S)=S\cup(S-1).
$
where \(S-1=\{s-1 \pmod{n}\mid s\in S\}\). Then
\begin{equation}\label{eq:pair-support}
w_{p}(\mathbf{c})=|N(\operatorname{supp}(\mathbf{c}))|.
\end{equation}
Consequently, if $\operatorname{supp}(\mathbf{c})\subseteq S$,
then
\begin{equation}\label{eq:pair-monotonicity}
w_{p}(\mathbf{c})\le |N(S)|.
\end{equation}
This monotonicity remains valid when zero coordinates
split one of the prescribed support intervals into
several nonzero runs.

\begin{lemma}\label{lem:basic}
The polynomial $g_\alpha(x)$ divides $x^n-1$ and has ten
distinct roots. Hence
\[
\dim_{\mathbb F_q}(\mathcal{C}_3)=n-10=2q-12.
\]
Moreover,
$
d_{\mathrm H}(\mathcal{C}_3)\ge8.
$
\end{lemma}

\begin{proof}
Since $q$ is odd and $\alpha$ is primitive,
$\alpha$ is a nonsquare in $\mathbb F_q$.
Choose $\beta\in\mathbb F_{q^2}$ with
$\beta^2=\alpha$. Then
\[
\beta^m=\alpha^{m/2}=-1,
\qquad
\operatorname{ord}(\beta)=2m=n,
\qquad
\beta^q=-\beta.
\]
The roots of $g_\alpha(x)$ are
$
1,\beta^2,\beta^4,\beta^6,
\pm\beta,\ \pm\beta^3,\ \pm\beta^5.
$
The first four are distinct elements of $\mathbb F_q$.
The remaining roots lie outside $\mathbb F_q$, since
their squares are odd powers of $\alpha$ and therefore
nonsquares.
Their constant terms $\alpha,\alpha^3,\alpha^5$ are
distinct because $\operatorname{ord}(\alpha)=m\ge10$.
Thus all ten roots are distinct. They are $n$-th roots
of unity, so  $g_\alpha(x)\mid x^n-1$.
The dimension follows from $\deg (g_\alpha(x))=10$.
Finally,
$
1,\beta,\beta^2,\ldots,\beta^6
$
are seven consecutive roots of $g_\alpha(x)$.
The BCH bound gives $d_{\mathrm H}(\mathcal{C}_3)\ge8$.
\end{proof}

We perform the same parity decomposition on the codeword, writing its polynomial representative as
$
c(x)=E(x^2)+xO(x^2),
$
where
\[
E(y)=\sum_{j=0}^{m-1}c_{2j}y^j,
\qquad
O(y)=\sum_{j=0}^{m-1}c_{2j+1}y^j.
\]
The conditions at $\pm\beta,\pm\beta^3,\pm\beta^5$
are equivalent to
\begin{equation}\label{eq:parity-roots}
\begin{aligned}
E(\alpha)=E(\alpha^3)=E(\alpha^5)&=0,\\
O(\alpha)=O(\alpha^3)=O(\alpha^5)&=0.
\end{aligned}
\end{equation}

\begin{lemma}\label{lem:parity-count}
Suppose that a support consists of cyclic runs of
total even length $L$. Let $r$ be the number of
odd length runs. Then $r$ is even, and the number
of even support positions lies between
\[
\frac{L-r}{2}
\quad\text{and}\quad
\frac{L+r}{2},
\]
inclusive.
\end{lemma}

\begin{proof}
Let $N_e$ denote the number of even indices in the support, and $N_o$ the number of odd indices in the support. Then $N_e + N_o = L$ is even. An even length segment contributes equally many even and odd positions, so $N_e - N_o$ equals the sum of contributions from all odd-length segments. Each odd length segment contributes either $+1$ or $-1$ to this difference. Since $N_e - N_o = L - 2N_o$ is even, the number $r$ of such segments must be even.
There are $r$ such segments in total. Each odd length segment can make $N_e$ exceed $N_o$ by at most $1$, or make $N_e$ fall below $N_o$ by at most $1$. Hence $-r \le N_e - N_o \le r$. Substitute $N_o = L - N_e$; then
$
-r \le N_e - (L - N_e) \le r,
$
and we thus obtain
\[
\frac{L - r}{2} \le N_e \le \frac{L + r}{2}.
\]
\end{proof}

\begin{lemma}
\label{lem:no-trinomial}
Let
$
F(y)=\sum_{i=1}^3 f_i y^{u_i},
$
where $f_i\in\mathbb F_q^*$ and the exponents $u_i$
are distinct modulo $m$. Then $F$ cannot vanish
simultaneously at $\alpha,\alpha^3,\alpha^5$.
\end{lemma}

\begin{proof}We proceed by contradiction. Assume that there exists a non zero polynomial \(F(y)\) such that
$F(\alpha)=F(\alpha^3)=F(\alpha^5)=0.$
Set
$
z_i=\alpha^{2u_i},
b_i=f_i\alpha^{u_i},
$ since $f_i\in\mathbb F_q^*$ and $\alpha\neq0$, all $b_i\neq0$.
The three vanishing equations give

\[
\begin{pmatrix}
1&1&1\\
z_1&z_2&z_3\\
z_1^2&z_2^2&z_3^2
\end{pmatrix}
\begin{pmatrix}b_1\\b_2\\b_3\end{pmatrix}=0.
\]
If the \(z_i\) are distinct, the matrix is invertible, the system of equations has only the trivial solution.
If exactly two of the \(z_i\) coincide, the coefficient of
the remaining column must be zero. Both cases
contradict $b_1b_2b_3\ne0$.
All three  \(z_i\) cannot coincide because each residue class
modulo $m/2$ contains only two residue classes modulo
$m$, whereas the $u_i$ are distinct modulo $m$.
\end{proof}

\subsection{Frobenius constraints on weight-eight supports}

\begin{lemma}\label{lem:frobenius-eight}
Let $\mathbf{c}\in \mathcal{C}_3$ have Hamming weight eight, with
four support positions of each parity. Define
\[
\mathcal E=\{\beta^s:s\in\operatorname{supp}(\mathbf{c}),\
s\text{ even}\},
\]
\[
\mathcal O=\{\beta^s:s\in\operatorname{supp}(\mathbf{c}),\
s\text{ odd}\},
\]
and
\[
A(z)=\prod_{e\in\mathcal E}(z-e)
=z^4-e_1z^3+e_2z^2-e_3z+e_4,
\]
\[
B(z)=\prod_{o\in\mathcal O}(z-o)
=z^4-o_1z^3+o_2z^2-o_3z+o_4.
\]
Then
\begin{equation}\label{eq:symmetric-constraints}
e_2=o_2,\qquad
e_4=o_4,\qquad
e_1o_3=e_3o_1.
\end{equation}
\end{lemma}

\begin{proof}
Put $H(z)=A(z)B(z)$. The codeword satisfies the root condition, so we have
\[
c(\beta^j)=0,\qquad 0\le j\le6,
\]
which forms  a $7\times8$ Vandermonde system with a
1-dimensional kernel. The standard interpolation
identity gives a kernel generator
$
\left(H'(\beta^s)^{-1}\right)_
{s\in\operatorname{supp}(c)}.
$
Thus, for some $\lambda\in\mathbb F_{q^2}^*$,
\[
c_s=\frac{\lambda}{H'(\beta^s)}.
\]
Since $c_s\in\mathbb F_q^*$, so
$
H'(\beta^s)^q=K H'(\beta^s),$ where $
K=\lambda^q/\lambda.
$

For each root $e\in E$ of $A(z)$, we have $e^q = e$. and for each root $o\in O$ of $B(z)$, we have $o^q = -o$.
Evaluate the identity $\big(H'(\gamma)\big)^q = K H'(\gamma)$ at the roots of $A(z)$ and $B(z)$.
For any $e\in E$, $H'(e)=A'(e)B(e)$. Since $A'(e)\neq 0$ and $A'(e)^q=A'(e)$, we obtain
$
B(e)^q = K B(e).
$
Using $B(e)^q = B(-e)$, we get $B(-e)-K B(e)=0$ for all four roots $e$ of $A(z)$.
Hence $B(-z)-K B(z)$ is divisible by $A(z)$. Comparing leading coefficients,
\[
B(-z)-K B(z) = (1-K)A(z).
\]
By the same argument, we get
$
A(-z)+K A(z) = (1+K)B(z).
$
The coefficients of $z^2$ and the constant terms
give $e_2=o_2$ and $e_4=o_4$, since the characteristic
is odd. The coefficients of $z^3$ and $z$ give
\[
(1-K)e_j+(1+K)o_j=0,\qquad j=1,3.
\]
Because $(1-K,1+K)\ne(0,0)$, the vectors
$(e_1,o_1)$ and $(e_3,o_3)$ are linearly dependent.
This proves the last identity.
\end{proof}

\subsection{The universal distance interval}

 \begin{theorem}\label{thm:interval}
Let $q\ge11$. Then
$\dim (\mathcal{C}_3)=2q-12,
$ and
$11\le d_{p}(\mathcal{C}_3)\le12.
$
\end{theorem}

\begin{proof}
The dimension follows from Lemma~\ref{lem:basic}.
The Singleton bound for symbol-pair codes gives
\[
d_{p}(\mathcal{C}_3)
\le n-\dim (\mathcal{C}_3)+2=12.
\]

Suppose that a nonzero word satisfies
$w_{p}(\mathbf{c})\le10$. Since $n\ge20$, its support
is proper. By $d_{\mathrm H}(\mathcal{C}_3)\ge8$ and
$wt_p(\mathbf{c})=w_H{(\mathbf{c})}+\rho(\mathbf{c})
$, we have
$
(w_{\mathrm H}(\mathbf{c}),\rho(\mathbf{c}))
\in\{(8,1),(8,2),(9,1)\}.
$
A one-run word of length at most nine can be shifted
to a nonzero polynomial of degree at most eight.
Such a polynomial cannot be divisible by
$g_\alpha$, which has degree ten.

It remains to exclude $(8,2)$. By
Lemma~\ref{lem:parity-count}, the parity counts are
$(3,5)$, $(4,4)$, or $(5,3)$. The first and last
possibilities contradict Lemma~\ref{lem:no-trinomial}.
Thus the split is $(4,4)$.
After a cyclic shift, write the support as
\[
S=\{0,\ldots,a-1\}
\cup\{t,\ldots,t+7-a\},
\qquad 1\le a\le4,
\]
with both cyclic gaps nonempty.
From the conclusion of Lemma ~\ref{lem:frobenius-eight}, \(e_4=o_4\). Recall that
\[e_4=\prod_{s\in S_{\text{even}}}\beta^s=\beta^{\sum_{s\in S_{\text{even}}}s},\qquad
o_4=\prod_{s\in S_{\text{odd}}}\beta^s=\beta^{\sum_{s\in S_{\text{odd}}}s}.\]
Equating them yields
\(\beta^{\sum_{s\in S_{\text{even}}}s-\sum_{s\in S_{\text{odd}}}s}=1.\)
Since \(\beta\) has order n, we obtain \(n\mid \Delta\), where
\(\Delta=\sum_{s\in S_{\text{even}}}s-\sum_{s\in S_{\text{odd}}}s.\)

The cases compatible with the parity split $(4,4)$
are
\[
\begin{array}{c|c|c}
a&\text{parity of }t&\Delta\\ \hline
1&\text{odd}&-(t+3)\\
2&\text{even}&-4\\
2&\text{odd}&2\\
3&\text{odd}&-(t+1)\\
4&\text{even}&-4\\
4&\text{odd}&0
\end{array}
\]
The nonempty-gap conditions imply
$0<|\Delta|<n$ in every row except the last.
Hence $a=4$ and $t=2s+1$.

For this remaining case,
\[
\mathcal E=\{1,\alpha,\alpha^{s+1},\alpha^{s+2}\},
\]
\[
\mathcal O
=\{\beta,\beta\alpha,\beta\alpha^s,\beta\alpha^{s+1}\}.
\]
By Lemma ~\ref{lem:frobenius-eight}, expanding $e_2-o_2$ gives
\[
e_2-o_2
=\alpha(\alpha-1)(\alpha^{2s+1}-1)=0.
\]
Thus $m\mid2s+1$, which is impossible because $m$
is even. This excludes every pair weight at most
ten.
\end{proof}

We now give two examples whose results are consistent
with Theorem~24.

\begin{example}
Let $q=19$ and choose $\alpha=3\in\mathbb F_{19}^{*}$,
which has order $18$.
The code $\mathcal{C}_3$ has generator polynomial
\[
\begin{aligned}
g_\alpha(x)={}&x^{10}+17x^9+3x^8+17x^7+12x^6+18x^4+14x^3+3x^2+3x+7
\end{aligned}
\]
and parameters $[36,26]_{19}$ with symbol-pair distance
$d_p(\mathcal{C}_3)=11$.
\end{example}

\begin{example}
Let $q=11$ and choose $\alpha=2\in\mathbb F_{11}^{*}$,
which has order $10$.
The code $\mathcal{C}_3$ has generator polynomial
\[
\begin{aligned}
g_\alpha(x)={}&x^{10}+7x^9+6x^8+4x^7+x^6+3x^4+8x^3+8x^2+5x+1
\end{aligned}
\]
and parameters $[20,10]_{11}$ with symbol-pair distance
$d_p(\mathcal{C}_3)=12$.
\end{example}

\noindent\textbf{Acknowledgments.} DeepSeek was used for improving the presentation of the paper.


\begin{thebibliography}{99}
\bibitem{1} Y. Cassuto and M. Blaum, Codes for Symbol-Pair Read Channels. IEEE Trans. Inf. Theory, 57(12): 8011-8020 (2011).



\bibitem{9} Y. M. Chee, L. Ji, H. M. Kiah, et al., Maximum Distance Separable Codes for Symbol-Pair Read Channels. IEEE Trans. Inf. Theory, 59(11): 7259-7267 (2013).





\bibitem{16} B. C. Chen, L. Lin, and H. Liu, Constacyclic Symbol-Pair Codes: Lower Bounds and Optimal Constructions. IEEE Trans. Inf. Theory, 63(12):7661-7666 (2017).

\bibitem{17} B. Ding, T. Zhang, and G. Ge, Maximum Distance Separable Codes for $b$-Symbol Read Channels. Finite Fields Appl., 49: 180-197 (2018).

\bibitem{13} H. Q. Dinh, B. T. Nguyen, A. K. Singh, et al., On the Symbol-Pair Distance of Repeated-Root Constacyclic Codes of Prime Power Lengths. IEEE Trans. Inf. Theory, 64(4): 2417-2430 (2017).


\bibitem{2} H. Q. Dinh, B. T. Nguyen, and S. Sriboonchitta, MDS Symbol-Pair Cyclic Codes of Length $2p^s$ over $\mathbb{F}_{p^m}$. IEEE Trans. Inf. Theory, 66(1): 240-262 (2019).



\bibitem{7} H. Q. Dinh, B. T. Nguyen, A. K. Singh, and S. Sriboonchitta, On Symbol-Pair Distances of Repeated-Root Constacyclic Codes of Length $2p^s$. Discrete Math., 342(11): 3062-3078 (2019).


    \bibitem{10} B. Ding, G. Ge, J. Zhang, et al., New Constructions of MDS Symbol-Pair Codes. Designs, Codes Cryptogr., 86(4): 841-859 (2018).


    \bibitem{11} X. S. Kai, S. X. Zhu, and P. Li, A Construction of New MDS Symbol-Pair Codes. IEEE Trans. Inf. Theory,  61(11): 5828-5834 (2015).


\bibitem{19} X. Kai, S. Zhu, Y. Zhao, et al., New MDS Symbol-Pair Codes From Repeated-Root Codes. IEEE Commun Letters, 22(3): 462-465 (2018).

\bibitem{22} X. Kai, Y. Zhou, and S. Zhu, Two New Classes of MDS Symbol-Pair Codes. IEEE Trans. Inf. Theory,  70(11):7701-7710 (2024).




\bibitem{20} G. Luo, M. F. Ezerman, S. Ling, and X. Pan, New Families of MDS Symbol-Pair Codes from Matrix-Product Codes. IEEE Trans. Inf. Theory, 69(3):1567-1587 (2023).

\bibitem{MS78} F. J. MacWilliams and N. J. A. Sloane, The Theory of Error-Correcting Codes. North-Holland Mathematical Library, Vol. 16, Amsterdam: North-Holland, (1977).


    \bibitem{15} J. Ma and J. Luo, MDS Symbol-Pair Codes from Repeated-Root Cyclic Codes. Designs, Codes Cryptogr., 90(1): 121-137 (2022).



\bibitem{21} J. Ma and J. Luo, Constructions of MDS Symbol-Pair Codes with Minimum Distance Seven or Eight. Designs, Codes Cryptogr., 90(10): 2337-2359 (2022).

\bibitem{25}J. Ma and J. Luo, AMDS symbol-pair codes from repeated-root cyclic codes. Discrete Math. 346, no. 7, Art. no. 113391 (2023).


\bibitem{23} R. Qiu and W. Fang, New Constructions of MDS Symbol-Pair Codes via Simple-Root Cyclic Codes. Designs, Codes Cryptogr., 93(11): 4857-4886 (2025).


\bibitem{26} X. Tang and R. Luo, MDS and AMDS symbol-pair codes constructed from repeated-root cyclic codes. Finite Fields Appl., 89, Art. no. 102203 (2023).








\end{thebibliography}
\end{document}